\documentclass{ifacconf}

\usepackage{graphicx}      
\usepackage{natbib}        

\usepackage{amsmath,amssymb,mathrsfs,theorem}
\usepackage{graphicx,array}
\usepackage{epsfig,color,algorithm,algorithmic}
\graphicspath{{./fig/}}

\newtheorem{theorem}{Theorem}[section]
\newtheorem{proposition}[theorem]{Proposition}

{\theorembodyfont{\rmfamily} 
\newtheorem{remark}[theorem]{Remark}

}

\newcommand{\VV}{\mathcal{V}}

\newcommand{\p}{\mathbf{p}}

\newcommand{\real}{{\mathbb{R}}}

\newcommand{\subscr}[2]{{#1}_{\textup{#2}}}

\newcommand{\map}[3]{#1: #2 \rightarrow #3}

\newcommand{\abs}[1]{|#1|}

\newcommand\oprocendsymbol{\hbox{$\square$}}
\newcommand\oprocend{\relax\ifmmode\else\unskip\hfill\fi\oprocendsymbol}

\begin{document}
\begin{frontmatter}

\title{Optimal Configurations in Coverage Control with Polynomial Costs\thanksref{footnoteinfo}} 

\thanks[footnoteinfo]{This
    work was supported by the United Technologies Research Center and ONR N00014-16-1-2722.}

\author[First]{Shaunak D. Bopardikar} 
\author[First]{Dhagash Mehta} 
\author[Second]{Jonathan D. Hauenstein}

\address[First]{United Technologies Research Center, 411 Silver Lane, East Hartford, CT 06118 USA (e-mails: \{bopardsd, mehtad\}@utrc.utc.com).}
\address[Second]{Department of Applied and Computational Mathematics and~Statistics, University of Notre Dame, Notre Dame, IN 46556 (e-mail: hauenstein@nd.edu).}

\begin{abstract}                
 We revisit the static coverage control problem for placement of vehicles with simple motion on the real line, under the assumption that the cost is a polynomial function of the locations of the vehicles. The main contribution of this paper is to demonstrate the use of tools from numerical algebraic geometry, in particular, a numerical polynomial homotopy continuation method
  that guarantees to find all solutions of polynomial equations, in order to characterize the \emph{global minima} for the coverage control problem. The results are then compared against a classic distributed approach involving the use of Lloyd descent, which is known to converge only to a local minimum under certain technical conditions.
\end{abstract}

\begin{keyword}
Coverage control, locational optimization, polynomial homotopy, numerical algebraic geometry.
\end{keyword}

\end{frontmatter}

\section{Introduction}
Vehicle placement to provide optimal coverage has received lot of
attention, especially in the past two decades. The goal is to determine where to place vehicles in order to optimize a specified cost that is a function of the locations of the vehicles. This paper addresses the characterization of the global minima for a vehicle placement problem under the assumption that this cost function is polynomial in the locations of the vehicles. It is well known that polynomials can be used as building blocks to describe several realistic functions. Applications of this work are envisioned in border patrol
wherein unmanned vehicles are placed to optimally intercept moving 
targets that cross a region under surveillance
(cf.~\cite{ARG-ASH-JKH:04, RS-MK-KL-DC:08}).

Vehicle placement problems are analogous to geometric location
problems, wherein given a set of static points, the goal is to find
supply locations that minimize a cost function of the distance from
each point to its nearest supply location (cf.~\cite{EZ:85}). For a
single vehicle, the expected distance to a point that is randomly
generated via a probability density function, is given by the
continuous $1$--median function.  The $1$--median function is 
minimized by a point termed as the \emph{median}
(cf.~\cite{SPF-JSBM-KB:05}). For multiple distinct vehicle locations,
the expected distance between a randomly generated point and one of
the locations is known as the continuous multi-median function
(cf.~\cite{ZD:95}). For more than one location, the multi-median
function is non-convex and thus determining locations that minimize
the multi-median function is hard in the general
case.~\cite{JC-SM-TK-FB:02j} addressed a distributed version of a
partition and gradient based procedure, known as the Lloyd algorithm,
for deploying multiple robots in a region to optimize a multi-median
cost function.~\cite{MS-DR-JJS:08} provided an adaptive control law to
enable robots to approximate the density function from sensor
measurements.~\cite{SM-FB:04p} presented motion coordination
algorithms to steer a mobile sensor network to an optimal
placement.~\cite{AK-SM:10} presented a coverage algorithm for vehicles
in a river environment. Related forms of the cost function have also
appeared in disciplines such as vector quantization, signal processing
and numerical integration (cf.~\cite{RMG-DLN:98, QD-VF-MG:99}).

In this paper, we consider the static coverage control problem for placement of vehicles with simple motion on the real line. We assume that the cost is a polynomial function of the locations of the vehicles. This structure implies that the set of all candidate optima is finite.
The main contribution of this paper is to demonstrate the use of tools from numerical algebraic geometry, 
in particular, a numerical polynomial homotopy continuation method that guarantees to find all solutions of the polynomial equations (cf. \cite{SWbook} and \cite{BertiniBook}).
Such methods have been used in a variety of 
problems, e.g., computing all finite and
infinite equilibria for constructing
one-dimensional slow invariant manifolds 
of dynamical systems (cf.~ 
\cite{SIM})
and finding all equilibria
of the Kuramoto model (cf.~\cite{KuramotoModel}).
Upon computing the finite set of candidate optima, 
we can evaluate the cost function at these points 
to obtain the \emph{global minimum}
for the coverage control problem. 
The results are then compared numerically using two examples with a classic distributed approach involving the use of Lloyd descent, which is known to converge only to a local minimum under certain technical conditions. We observe that in one of the examples, both methods lead to the same global minimizer, while in the second example, the Lloyd descent converges to only a local minimum if initialized from particular configurations.

This paper is organized as follows. The problem is formulated in
Section~\ref{secn:prob}. The multiple vehicle scenario is addressed
in Section~\ref{secn:multi}. The polynomial homotopy method is reviewed and its application to the coverage problem is presented in Section~\ref{sec:homotopy}. Numerical simulation results are presented in Section~\ref{sec:simulations}. 

\section{Problem Statement}\label{secn:prob}

We consider vehicles modeled with single integrator dynamics having
unit speed. A static target is generated at a random position $x \in [A,B]$ on the
segment $G:=[A, B]$, via a specified probability density function \mbox{$\map{\phi}{[A,B]}{\real_{\geq
    0}}$}. We assume that the density $\phi$ is bounded, and therefore integrable over a compact domain. The goal is to determine vehicle placements that minimize the expected time for the nearest vehicle to reach a target. We consider the both the single and multiple vehicle 
cases.
\subsection{Single Vehicle Case}\label{secn:probsingle}
We determine a vehicle location $p\in [A,B]$ that
minimizes $\map{\subscr{C}{exp}}{[A,B]}{\real}$ given by
\begin{equation}\label{eq:cost}
\subscr{C}{exp}(p):= \int_A^B C(p,x)\phi(x)dx, 
\end{equation}
where $\map{C}{[A,B]\times[A,B]}{\real_{\geq 0}}$
is an appropriately defined cost of the vehicle position $p$ and the target location $x$. In what
follows, we consider costs with the following properties.

(i) \emph{Polynomial dependence on $p$:} We assume that for any $p \in [A,B]$, the cost function $C$ is polynomial in $p$.


(ii) \emph{Homogeneity:} We assume that the function $C$ satisfies
\[
C(p, x) = \frac{1}{2}f((p-x)^2),
\]
where $f(\cdot) \geq 0$ is a polynomial and is monotonic with respect to its argument.

\subsection{Multiple Vehicles Case} \label{secn:probmulti}
Given $m\geq 2$
vehicles, the goal is to determine a set of vehicle
locations $p_i$, for every $i\in
\{1,\dots,m\}$, that minimizes the expected cost given by
\begin{equation}\label{eq:multtime}
\subscr{C}{exp}(p_1,\dots,p_m) :=
\int_A^B\min_{i\in\{1,\dots,m\}}C(p_i,x)\phi(x)dx,
\end{equation}
where $C(p_i,x)$ satisfies the same properties that are assumed in Section~\ref{secn:probsingle}. The single vehicle case shall then follow as a special case of multiple vehicles.

\section{The Case of Multiple Vehicles} \label{secn:multi}

Consider the multiple vehicle case from Section~\ref{secn:probmulti} with 
assumptions (i) and (ii) from Section~\ref{secn:probsingle}.
We will require the concept of \emph{dominance} regions. For the $i$-th vehicle, the dominance region $\VV_i$ is defined as
\[
\VV_i := \{x \in [A,B]\, :\, C(p_i,x) \leq C(p_j,x), \forall j \neq i \}.
\]
In other words, $\VV_i$ is the set of all points $x$ for which an assignment of any point in that set to vehicle $i$ provides the least cost over assignment to any other vehicle. The following proposition provides a simple approach to computing the dominance regions.

\begin{proposition}
Under Assumption (ii), the dominance region of the $i$-th vehicle is the Euclidean Voronoi partition corresponding to the $i$-th vehicle, i.e.,
\[
\VV_i = \{x \in [A,B] \, : \, \abs{p_i - x} \leq \abs{p_j-x}, \forall j \neq i\}.
\]
\end{proposition}
\begin{proof}
From the definition of $\VV_i$, we have
\begin{multline*}
 C(p_i,x) \leq C(p_j,x) \\ \Rightarrow (p_i - x)^2 \leq (p_j-x)^2 \Rightarrow \abs{p_i-x} \leq \abs{p_j-x},
\end{multline*}
from the monotonicity in Assumption (ii). 
\end{proof}
Without any loss of generality, let the vehicles be placed with their indices in ascending order on $[A,B]$. Then,
\begin{align*}
\VV_i = \begin{cases} [A, (p_1 + p_2)/2], &i = 1, \\
[(p_{i-1}+p_i)/2, (p_i + p_{i+1})/2], &i \in \{2, \dots, m-1\}, \\
[(p_{m-1}+p_m)/2, B], &i = m.
\end{cases}
\end{align*}

\subsection{Minimizing the Expected Constrained Travel Time} \label{secn:exptime}
For distinct vehicle locations, \eqref{eq:multtime} can be written
as
\begin{equation}\label{eq:texp}
\subscr{C}{exp}(p_1,\dots,p_m) = \sum_{i=1}^m\int_{\VV_i}C(p_i,x)\phi(x)dx,
\end{equation}
where $\VV_i$ is the dominance region of the $i$-th vehicle. The
gradient of $\subscr{C}{exp}$ is computed using the following formula,
which allows each vehicle to compute the gradient of
$\subscr{C}{exp}$ by integrating the gradient of $C$ over
$\VV_i$. 

\begin{proposition}[Gradient computation]\label{lem:grad_exptime}
  For all vehicle configurations such that no two vehicles are at
  coincident locations, the gradient of the expected time with respect to
  vehicle location $p_i$ is
  \[
  \frac{\partial \subscr{C}{exp}}{\partial p_i} = \int_{\VV_i}\frac{\partial C}{\partial p_i}(p_i,x)\phi(x)dx. 
  \]
\end{proposition}

\begin{proof}
Akin to similar results in~\cite{FB-JC-SM:09}, 
the following involves
writing the gradient of $\subscr{C}{exp}$ as a sum of two contributing
terms. The first is the final expression, while the second is a number
of terms which cancel out due to continuity of $C$ at the boundaries
of dominance regions.

 Let $p_j$ be termed as a \emph{neighbor} of $p_i$, i.e., $j\in$
  neigh$(i)$, if $\VV_i\cap\VV_j$ is non-empty. Then,
\begin{multline}\label{eq:gradientsum}
\frac{\partial \subscr{C}{exp}}{\partial p_i}  = \frac{\partial}{\partial p_i}\int_{\VV_i}C(p_i,x)\phi(x)dx \\ + \sum_{j\text{ neigh }(i)}\frac{\partial}{\partial p_i}\int_{\VV_j}C(p_j,x)\phi(x) dx,
\end{multline}
Now, from the expression of $\VV_i$, there arise three cases:

1. \emph{$i \in \{2,\dots, m-1\}$}: In this case, all
boundary points $(p_{i-1}+p_i)/2$ and $(p_i+p_{i+1})/2$ are differentiable with respect to
$p_i$. Therefore, by Leibnitz's Rule\footnote{
$\frac{\partial}{\partial z}\int_{a(z)}^{b(z)} f(z,x)dx = \int_{a(z)}^{b(z)} \frac{\partial f(z,x)}{\partial z}dx + f(z,b)\frac{\partial b(z)}{\partial z} - f(z,a)\frac{\partial a(z)}{\partial z}$},
\begin{align*}
&\frac{\partial}{\partial p_i}\int_{\VV_i}C(p_i,x)\phi(x)dx - \int_{\VV_i}\frac{\partial C}{\partial p_i}\phi(x)dx
\\ &= \frac{1}{2} \left(C\left(p_i, \frac{p_i+p_{i+1}}{2}\right) - C\left(p_i,\frac{p_i+p_{i-1}}{2}\right)\right) \\
&= \frac{1}{2} 
\left(f\left(\left(\frac{p_{i+1}-p_i}{2}\right)^2\right) - 
      f\left(\left(\frac{p_i-p_{i-1}}{2}\right)^2\right)\right),
\end{align*}
where the second step follows from Assumption (ii).

Using the same steps by applying the Leibnitz rule, we conclude that
\begin{align*}
\frac{\partial}{\partial p_i}\int_{\VV_{i-1}}C(p_{i-1},x)\phi(x)dx &=  \frac{1}{2} f\left(\left(\frac{p_{i}-p_{i-1}}{2}\right)^2\right) \\
\frac{\partial}{\partial p_i}\int_{\VV_{i+1}}C(p_{i+1},x)\phi(x)dx &= -\frac{1}{2} f\left(\left(\frac{p_{i+1}-p_{i}}{2}\right)^2\right). 
\end{align*}

Therefore, combining these three expressions into~\eqref{eq:gradientsum}, we conclude that for this case,
 \[
  \frac{\partial \subscr{C}{exp}}{\partial p_i} = \int_{\VV_i}\frac{\partial C}{\partial p_i}(p_i,x)\phi(x)dx. 
  \]

2. \emph{$i = 1$}: In this case, the lower limit of the integral below is $A$ and is a constant. Therefore, applying the Leibnitz rule, we obtain
{\small
\[
\frac{\partial}{\partial p_1}\int_{\VV_1}C(p_1,x)\phi(x)dx = \int_{\VV_1}\frac{\partial C}{\partial p_1}\phi(x)dx 
+  \frac{1}{2} f\left(\left(\frac{p_{2}-p_1}{2}\right)^2\right). 
\]}Applying the Leibnitz rule for the corresponding term involving the neighbor $p_2$, we obtain
\[
\frac{\partial}{\partial p_1}\int_{\VV_{2}}C(p_{2},x)\phi(x)dx =  -\frac{1}{2} f\left(\left(\frac{p_{2}-p_{1}}{2}\right)^2\right) 
\]

Therefore, combining these two expressions into~\eqref{eq:gradientsum}, we conclude that for this case,
 \[
  \frac{\partial \subscr{C}{exp}}{\partial p_1} = \int_{\VV_1}\frac{\partial C}{\partial p_1}(p_1,x)\phi(x)dx. 
  \]

The third case of $i=m$ is very similar to the case of $i=1$ and the conclusion analogous to that of $i=1$ can be verified. Therefore, the claim is verified for 
each case. 
\end{proof}

\begin{remark}[Infinite interval]
The formulation can easily be extended to the case when the domain for the vehicles is unbounded, i.e., $\real$. In that case, we will require an extra assumption on the weight function $\phi$ which would be that as $x\to \pm \infty$, $\phi(x) \to 0^+$ while $C(x)$ remains bounded.
\end{remark}

The expressions for the gradient can then be used within the Lloyd descent algorithm (see for example~\cite{FB-JC-SM:09}) to derive a control scheme for each vehicle to move, beginning with an initial arbitrary, non-degenerate configuration using the following steps iteratively: while a given number of iterations are not reached,
\begin{enumerate}
\item Each vehicle computes its Voronoi partition $\VV_i$,
\item Each vehicle computes the gradient of the cost function using Proposition~\ref{lem:grad_exptime},
\item Each vehicle computes its step size using backtracking line search, and
\item Each vehicle uses gradient descent to compute its new position.
\end{enumerate}

In the following subsection, we will address computing the global optima by posing the set of equations that need to be solved in order to compute the candidate points.

\subsection{Optimal Placement}
The optimal vehicle placement problem is cast as
\begin{align*}
&\min_{\{p_1,\dots, p_n\}\in [A,B]^m} \subscr{C}{exp}(p_1,\dots,p_m) \\
&\text{subject to } p_i \in [A,B], \forall i \in \{1,\dots, m\}.
\end{align*}

Without loss of generality, we assume that the vehicles are located such that $p_{i-1} < p_i < p_{i+1}$. Then, the candidate global minima are:
\begin{enumerate}
\item $p_1^* = A$ and the set of all points $p_i^*, i=2,\dots,m$, for which
\[
 \frac{\partial \subscr{C}{exp}}{\partial p_i} (p_1^*, \dots, p_m^*) = 0,
\]
\item $p_m^* = B$ and the set of all points 
$p_i^*, i = 1,\dots, m-1$, for which
\[
 \frac{\partial \subscr{C}{exp}}{\partial p_i} (p_1^*, \dots, p_m^*) = 0,
\]
or,
\item the set of all points $p_i^*, i = 1,\dots,m$, for which
\[
 \frac{\partial \subscr{C}{exp}}{\partial p_i} (p_1^*, \dots, p_m^*) = 0,
\]
\end{enumerate}
along with the additional condition on the Hessian
\[
 \frac{\partial^2 \subscr{C}{exp}}{\partial \p^2} (p_1^*, \dots, p_m^*) \succ 0,
\]
where $\p := [p_1, \dots, p_m]$. Proposition~\ref{lem:grad_exptime} provides a simple expression for the computation of the partial derivatives of $\subscr{C}{exp}$. The set of all candidate can be characterized by
\begin{align*}
\int_A^{\frac{p_1+p_2}{2}} f'\left((p_1-x)^2\right)(p_1-x)\phi(x)dx &= 0, \\
\int_{\frac{p_{m-1}+p_m}{2}}^B f'\left((p_m-x)^2\right)(p_m-x)\phi(x)dx &= 0,
\end{align*} 
and, for $2\leq i \leq m-1$, 
\begin{align*}
\int_{\frac{p_{i-1}+p_i}{2}}^{\frac{p_{i}+p_{i+1}}{2}}f'\left((p_i-x)^2\right)(p_i-x)\phi(x)dx &= 0
\end{align*}

Now let us call the integral
\[
F(p,b, a) := \int_a^b f'\left((p-x)^2\right)(p-x)\phi(x)dx
\]
Notice that under Assumption (ii), $F(p, b, a)$ is also a polynomial in $p$. Then, the candidates
for global minima are given by the set of polynomial equations:
\begin{align}\label{eq:optima}
F\left(p_1,\frac{p_1+p_2}{2},A\right) &= 0 \nonumber \\
F\left(p_i,\frac{p_{i+1}+p_i}{2}, \frac{p_{i}+p_{i-1}}{2}\right) &= 0 \quad \forall i \in \{2,\dots, m-1\}, \nonumber \\
F\left(p_m, B, \frac{p_{m-1} + p_m}{2}\right) &= 0,
\end{align} 
with the additional possibility that 
$p_1 = A$ or \mbox{$p_m = B$}. In the next section, we will review techniques from numerical algebraic geometry to explore the 
full spectrum of solutions for \eqref{eq:optima} and therefore compute the global optimum.

\section{Polynomial system of equations through Algebraic Geometry}\label{sec:homotopy}

Typically, performing an exhaustive search of solutions of systems of nonlinear equations such as \eqref{eq:optima}
is a prohobitively difficult task. 
However, in this paper, by restricting ourselves to 
polynomial conditions, this becomes feasible.
Furthermore, though in the original formulation
only requires computing real solutions of the system, we expand our search space to complex space, 
i.e., instead of $\textbf{p} \in \mathbb{R}^{n}$
we take $\textbf{p} \in \mathbb{C}^n$. 
The purpose of the complexification of the variables is to enable us to use some of the powerful 
mathematical and computational tools from algebraic geometry, i.e., in mathematical terms, 
$\mathbb{C}^n$ is the algebraic closure of $\mathbb{R}^n$.  In particular, we utilize the numerical
algebraic geometric computational technique
called the numerical polynomial homotopy 
continuation~(NPHC) method which guarantees (in the probability 1 sense) to compute all complex isolated solutions of a 
well-constrained system of multivariate polynomial equations.  More details are provided
in the books by \cite{SWbook} and 
\cite{BertiniBook}.

For a well-constrained system of polynomial equations 
(also called a square system which has the same number of 
equations and variables) $\textbf{F}(\textbf{p}) = \textbf{0}$, classical NPHC method 
uses a single homotopy that 
starts with an upper bound on the number
of isolated complex solutions.  
One standard upper bound is the classical
B\'ezout bound (CBB) which is simply
the product of the degrees of the polynomials, 
namely $\prod_{i=1}^{n} d_i$ where $d_i = \deg {\bf F}_i$
and $n$ is the number of polynomials in ${\bf F}$.
Although the CBB is trivial to compute, 
it does not take structure (such as sparsity or sparsity) of the system into account.
There are tighter bounds such as the 
multihomogeneous B\'ezout bound
and the polyhedral, also called the
Bernshtein-Kushnirenko-Khovanskii (BKK), bound
can exploit some structure in the system
to provide a tighter upper bound
using possibly
significant additional computations.

Each such upper bound yields a corresponding system 
$\textbf{G}(\textbf{p}) = \textbf{0}$,
called a start system,
where the bound is sharp.  For example,
the CBB yields 
$$\textbf{G}(\textbf{p}) = [p_1^{d_1}-1, \dots, p_{n}^{d_n} -1] = {\bf 0}$$
which clearly has $\prod_{i=1}^n d_i$ isolated solutions.
For other upper bounds, the procedure of constructing a start system may be more involved. 
Once a start system is constructed, 
a homotopy between $\textbf{F}(\textbf{p})$ and $\textbf{G}(\textbf{p})$ is constructed as
\begin{equation*}
 \textbf{H}(\textbf{p},t) = (1-t) \textbf{F}(\textbf{p}) + e^{\theta\sqrt{-1}} \, t \, \textbf{G}(\textbf{p}) = \textbf{0}
\end{equation*}
where $\theta\in[0,2\pi)$.  
One tracks the solution path defined by
${\bf H} = {\bf 0}$ from 
a known solution of ${\bf G} = {\bf 0}$
at $t = 1$ to $t = 0$.  
For all but finitely many $\theta\in[0,2\pi)$,
all solution paths are smooth for $t\in(0,1]$
and the set of isolated solutions
of ${\bf F} = {\bf 0}$ is contained
in the set of limit points of
the paths that converge at $t\rightarrow0^+$.
Since each path can be tracked independent of each other, the NPHC method is embarrassingly parallelizeable.

To exploit some of the structure in the 
system \eqref{eq:optima}, we first 
factor each polynomial ${\bf F}_i$,
say ${\bf F}_i = q_{i1}^{r_{i1}}\cdots q_{i k_i}^{r_{i k_i}}$ where $q_{ij}$ are polynomials
and $r_{ij}$ are positive integers.
Thus, we can replace each ${\bf F}_i$ 
with a square-free factorization
$q_{i1}\cdots q_{i k_i}$
to remove trivial singularities caused
by $r_{ij} > 1$ thereby improving
the numerical conditioning of the homotopy
paths.  Thus, the solutions of
${\bf F} = {\bf 0}$ is equal to the union
of the solutions of
$${\bf Q}_{j_1,\dots,j_n} = [q_{1j_1},\dots,q_{nj_n}] = {\bf 0}$$
where $1 \leq j_i \leq k_i$ for $i = 1,\dots,n$.
There are several benefits from such an approach.
First, we again improve numerical conditioning
of the homotopy paths by solving
lower degree systems and with the
removal of trivial singularities
that arise from solutions that simultaneously 
solve two or more such systems.
Second, this produces additional parallelization
opportunity, e.g., by solving 
each subsystem independently.
Rather than having a completely independent 
solving, we could utilize ideas
of regeneration developed by 
\cite{Regeneration} based on
bootstrapping from solving subsystems.  
To highlight the potential, suppose
that one has found that the subsystem
${\bf Q}_{j_1,\dots,j_s} = [q_{1j_1},\dots,q_{s j_s}] = {\bf 0}$ 
has no solutions, then one immediately knows
that ${\bf Q}_{j_1,\dots,j_s,j_{s+1},\dots,j_n} = {\bf 0}$ has no solutions for every
$1\leq j_i \leq k_i$ for $i = s+1,\dots,n$.  

For the regeneration approach, 
fix indices $j_1,\dots,j_n$ such that $1\leq j_i \leq k_i$ and we aim to solve ${\bf Q}_{j_1,\dots,j_n} = {\bf 0}$.  
Let $\ell_{iu}$ be general linear polynomials
for $i = 1,\dots,n$ and $u = 1,\dots,D_{i}$
where $D_{i} = \deg q_{ij_i}$.  
Consider the polynomial systems 
$${\bf G}^s_{u_{s+1},\dots,u_n} = 
[q_{1j_1},\dots,q_{s j_s},\ell_{s+1,u_{s+1}},\dots,\ell_{nu_n}].$$
Using linear algebra, we solve 
${\bf G}^0_{1,\dots,1} = {\bf 0}$.  
Regeneration using a two-stage approach
to use the solutions, say~$S_1$, of 
${\bf G}^s_{1,\dots,1} ={\bf 0}$
to compute the solutions of 
${\bf G}^{s+1}_{1,\dots,1} ={\bf 0}$ as follows.  
The first stage, 
for each $u = 2,\dots,D_{s+1}$, uses
the~homotopy
$$
(1-t){\bf G}^s_{u,1,\dots,1} 
+ t{\bf G}^s_{1,\dots,1} ={\bf 0} 
$$
with start points $S_1$ at $t = 1$
to compute the solutions $S_u$ of
${\bf G}^s_{u,1,\dots,1} = {\bf 0}$
at $t=0$.  By genericity, 
every path in this homotopy is smooth
for $t\in[0,1]$ so that $\# S_1 = \# S_u$.
Let $S = \cup_{i=1}^{D_{s+1}} S_i$ 
which are the solutions of
$$
{\bf K}^s_{1,\dots,1}
= \left[q_{1j_1},\dots,q_{sj_s},
\prod_{\alpha=1}^{D_{s+1}} \ell_{s+1,\alpha},
\ell_{s+2,1},\dots,\ell_{n1}\right] = {\bf 0}.
$$
Then, the second stage uses the homotopy
$$
(1-t) {\bf G}^{s+1}_{1,\dots,1} + 
t {\bf K}^s_{1,\dots,1} = {\bf 0}
$$
with start points $S$ at $t = 1$
to compute the solutions of
${\bf G}^{s+1}_{1,\dots,1} = {\bf 0}$
as desired.  Iterating this process 
produces the solutions 
to ${\bf Q}_{j_1,\dots,j_n} = {\bf 0}$.

A regeneration-based approach
can be advantageous over using a
single homotopy when the subsystems
have far fewer solutions than the selected
upper bound would predict.  
When the upper bound is not sharp,
this causes paths to diverge to infinity
resulting in wasted extra computation.  
This is reduced in regeneration 
by performing a sequence
of homotopies.  Likewise,
if the upper bound is sharp, then 
regeneration is not advantageous
due to the extra tracking through this sequence.
Nonetheless, we can produce all isolated
complex solutions using 
either the classic single homotopy
or regeneration which allows us
to always compute the global optimum.

\section{Simulations} \label{sec:simulations}
In this section, we numerically evaluate the proposed method and compare it with the classic Lloyd descent algorithm for two choices of the 
density functions $\phi$. The first scenario is selected in such a way that, for almost all initial configurations, the Lloyd algorithm tends to the global optimum. The second scenario is one in which, from some special initial conditions, the Lloyd algorithm tends to only a particular local optimum.

For the first scenario, we take 
the interval $[A,B]$ to be~$[0,W]$ for $W > 0$, 
$C(p,x) = f\left((p-x)^2\right) = (p-x)^2$,
and the weight function $\phi(x) := x(W-x)$. 
Then, we have
\begin{multline*}
6F(p, b, a) = 12\int_a^b (p-x) x(W-x) dx = \\ 
\Big( b^2 (6 p W - 4 p b - 4 W b + 3 b^2) - a^2 (6 p W - 4 p a - 4 W a + 3 a^2) \Big).
\end{multline*}

Substituting into the system of equations~\eqref{eq:optima}, we obtain 

{\tiny
\begin{align}
\nonumber&\Big( \frac{p_1 + p_2}{2} \Big )^2 \Big (6 p_1 W - 4 p_1 \Big( \frac{p_1 + p_2}{2} \Big ) - 4 W \Big( \frac{p_1 + p_2}{2} \Big ) + 3 \Big( \frac{p_1 + p_2}{2} \Big )^2 \Big )  = 0, \\
& \Big( \frac{p_{3} + p_{2}}{2} \Big )^2 \Big (6 p_2 W - 4 p_2  \Big( \frac{p_3 + p_2}{2} \Big ) - 4W \Big( \frac{p_{3} + p_2}{2} \Big ) + 3 \Big( \frac{p_{3} + p_2}{2} \Big )^2 \Big )  \nonumber \\ 
\nonumber&=  \Big( \frac{p_2 + p_{1}}{2} \Big )^2 \Big (6 p_2 W - 4 p_2 \Big( \frac{p_2 + p_{1}}{2} \Big ) - 4W \Big( \frac{p_2 + p_{1}}{2} \Big ) + 3 \Big( \frac{p_2 + p_{1}}{2} \Big )^2 \Big), \\
& W^2 \Big (6 p_3 W - 4 p_3 W - 4 W^2 + 3 W^2 \Big )  \nonumber \\ 
&=  \Big( \frac{p_2 + p_3}{2} \Big )^2 \Big (6 p_3 W - 4 p_3  \Big( \frac{p_2 + p_3}{2} \Big ) - 4 W \Big( \frac{p_2 + p_3}{2} \Big ) + 3 \Big( \frac{p_2 + p_3}{2} \Big )^2 \Big).\label{eq:System1}
\end{align}
}The system \eqref{eq:System1} consists of
three polynomial equations in three unknowns $p_1, p_2, p_3$ with the added possibilities of $p_1 = 0$
or $p_3 = W$.  
Taking $W = 1$, solving all three possibilities
following Section~\ref{sec:homotopy}
yields a total of 44 solutions in $\mathbb{C}^3$,
of which~$32$ are in~$\mathbb{R}^3$.  
Testing the objective function
yields that the global optimal solution 
is approximately~$(0.235, 0.5, 0.765)$. 
For randomly generated initial vehicle locations, we applied the well-known Lloyd descent algorithm and observe that the vehicle locations converge to the same configuration as illustrated in Figure~\ref{fig:symmetric}.

\begin{figure}[!h]
\centering
\includegraphics[width=\columnwidth]{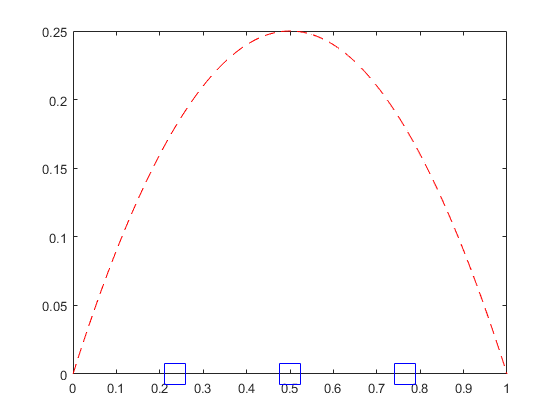}
\caption{Final configuration of vehicle locations obtained by running the Lloyd descent algorithm. The red dashed line shows the density function $\phi(x) = x(1-x)$, while the blue squares denote the three vehicles.}
\label{fig:symmetric}
\end{figure}

In the second scenario, 
we consider the weight function $\phi(x) := x^2-x^4$ in the interval $[A, B] = [-W, W]$. 
Then, the polynomial $F(p, b, a)$ is given by
\begin{multline*}
F(p, b, a) = \int_a^b (p-x) (x^2-x^4) dx \\ = \frac{p}{3}(b^3 - a^3) - \frac{1}{4}(b^4 - a^4) - \frac{p}{5}(b^5 - a^5) + \frac{1}{6}(b^6 - a^6). 
\end{multline*}
Substituting into the system of equations~\eqref{eq:optima}, for the interval $[-W, W]$, we obtain {\tiny
\begin{align}
\nonumber&\frac{p_1}{3}\Big(\Big (\frac{p_1 + p_2}{2} \Big)^3 + W^3\Big) - \frac{1}{4}\Big(\Big (\frac{p_1 + p_2}{2} \Big)^4 - W^4 \Big) \nonumber \\  
\nonumber &- \frac{p_1}{5}\Big( \Big (\frac{p_1 + p_2}{2} \Big)^5 + W^5 \Big)  + \frac{1}{6}\Big(\Big (\frac{p_1 + p_2}{2} \Big)^6 - W^6 \Big) = 0, \\
&\frac{p_2}{3}\Big(\Big (\frac{p_2 + p_3}{2} \Big)^3 - \Big (\frac{p_1 + p_2}{2} \Big)^3\Big) - \frac{1}{4}\Big(\Big (\frac{p_2 + p_3}{2} \Big)^4 - \Big (\frac{p_1 + p_2}{2} \Big)^4\Big) \nonumber \\
\nonumber&- \frac{p_2}{5}\Big(\Big (\frac{p_2 + p_3}{2} \Big)^5 - \Big (\frac{p_1 + p_2}{2} \Big)^5\Big) + \frac{1}{6}\Big(\Big (\frac{p_2 + p_3}{2} \Big)^6 - \Big (\frac{p_1 + p_2}{2} \Big)^6\Big) = 0,\\
&\frac{p_3}{3}\Big(W^3 - \Big (\frac{p_2 + p_3}{2} \Big)^3\Big) - \frac{1}{4}\Big(W^4 - \Big (\frac{p_2 + p_3}{2} \Big)^4\Big)\nonumber \\ &- \frac{p_3}{5}\Big(W^5 - \Big (\frac{p_2 + p_3}{2} \Big)^5 \Big) + \frac{1}{6}\Big(W^6 - \Big (\frac{p_2 + p_3}{2} \Big)^6\Big) = 0.\label{eq:System2}
\end{align}
}

The system \eqref{eq:System2} also consists
of three polynomial equations in three 
unknowns $p_1, p_2, p_3$. Using the technique from Section~\ref{sec:homotopy}, we obtain 
a total of 122 solutions in $\mathbb{C}^3$,
of which~30 are in $\mathbb{R}^3$.
This computation yields that the 
global optimal solution
is approximately $(-0.626,0.431,0.762)$. 
However, for this problem, from a class of initial configurations of the type $(-a, 0, a)$, where $a \in (0,W)$, the Lloyd descent algorithm leads the vehicles toward a local minimum approximately at
$(-0.66,0,0.66)$ as illustrated in Figure~\ref{fig:higherorder}.


\begin{figure}[!h]
\centering
\includegraphics[width=\columnwidth]{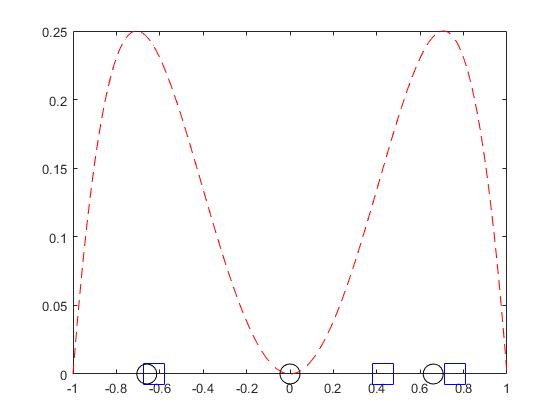}
\caption{Final configuration of vehicle locations obtained by running the Lloyd descent algorithm. The red dashed line shows the density function $\phi(x) = -x^4+x^2$. The black circles denote the final output of the Lloyd algorithm when initialize from (-a, 0, a), while the blue squares denote the final output when the vehicles are initialized with a random initial configuration.}
\label{fig:higherorder}
\end{figure}

\section{Conclusions and Future Directions}

This paper considered the static coverage control problem for placement of vehicles with simple motion on the real line. We assumed that the cost is a polynomial function of the locations of the vehicles. Our main contribution was to demonstrate the use of a numerical polynomial homotopy continuation method that guarantees to find all solutions of polynomial equations, in order to characterize the \emph{global minima} for the coverage control problem. The results were compared numerically using two examples with a classic distributed approach involving the use of Lloyd descent, known to converge only to a local minimum under certain technical conditions. We observed that in one of the examples, both methods lead to the same global minimizer, while in the second example, the Lloyd descent converges to only a local minimum when initialized from a particular class of configurations.

Future work is expected to center around fully distributed implementations of the polynomial homotopy method based on exploiting the structure of polynomials. We also plan to explore the complexity of this 
polynomial homotopy method 
in higher dimensional spaces. 

\bibliography{references} 

\begin{thebibliography}{17}
\providecommand{\natexlab}[1]{#1}
\providecommand{\url}[1]{\texttt{#1}}
\expandafter\ifx\csname urlstyle\endcsname\relax
  \providecommand{\doi}[1]{doi: #1}\else
  \providecommand{\doi}{doi: \begingroup \urlstyle{rm}\Url}\fi

\bibitem[Al-Khateeb et~al.(2009)Al-Khateeb, Powers, Paolucci, Sommese, Diller,
  Hauenstein, and Mengers]{SIM}
Ashraf~N. Al-Khateeb, Joseph~M. Powers, Samuel Paolucci, Andrew~J. Sommese,
  Jeffrey~A. Diller, Jonathan~D. Hauenstein, and Joshua~D. Mengers.
\newblock One-dimensional slow invariant manifolds for spatially homogenous
  reactive systems.
\newblock \emph{The Journal of Chemical Physics}, 131\penalty0 (2):\penalty0
  024118, 2009.

\bibitem[Bates et~al.(2013)Bates, Hauenstein, Sommese, and
  Wampler]{BertiniBook}
Daniel~J. Bates, Jonathan~D. Hauenstein, Andrew~J. Sommese, and Charles~W.
  Wampler.
\newblock \emph{Numerically solving polynomial systems with {B}ertini},
  volume~25 of \emph{Software, Environments, and Tools}.
\newblock Society for Industrial and Applied Mathematics (SIAM), Philadelphia,
  PA, 2013.

\bibitem[Bullo et~al.(2009)Bullo, Cortes, and Martinez]{FB-JC-SM:09}
Francesco Bullo, Jorge Cortes, and Sonia Martinez.
\newblock \emph{Distributed control of robotic networks: a mathematical
  approach to motion coordination algorithms}.
\newblock Princeton University Press, 2009.

\bibitem[Cortes et~al.(2004)Cortes, Martinez, Karatas, and
  Bullo]{JC-SM-TK-FB:02j}
Jorge Cortes, Sonia Martinez, Timur Karatas, and Francesco Bullo.
\newblock Coverage control for mobile sensing networks.
\newblock \emph{IEEE Transactions on robotics and Automation}, 20\penalty0
  (2):\penalty0 243--255, 2004.

\bibitem[Drezner and Hamacher(2001)]{ZD:95}
Zvi Drezner and Horst~W Hamacher.
\newblock \emph{Facility location: applications and theory}.
\newblock Springer Science \& Business Media, 2001.

\bibitem[Du et~al.(1999)Du, Faber, and Gunzburger]{QD-VF-MG:99}
Qiang Du, Vance Faber, and Max Gunzburger.
\newblock Centroidal voronoi tessellations: Applications and algorithms.
\newblock \emph{SIAM review}, 41\penalty0 (4):\penalty0 637--676, 1999.

\bibitem[Fekete et~al.(2005)Fekete, Mitchell, and Beurer]{SPF-JSBM-KB:05}
S{\'a}ndor~P Fekete, Joseph~SB Mitchell, and Karin Beurer.
\newblock On the continuous fermat-weber problem.
\newblock \emph{Operations Research}, 53\penalty0 (1):\penalty0 61--76, 2005.

\bibitem[Girard et~al.(2004)Girard, Howell, and Hedrick]{ARG-ASH-JKH:04}
Anouck~R Girard, Adam~S Howell, and J~Karl Hedrick.
\newblock Border patrol and surveillance missions using multiple unmanned air
  vehicles.
\newblock In \emph{Decision and Control, 2004. CDC. 43rd IEEE Conference on},
  volume~1, pages 620--625. IEEE, 2004.

\bibitem[Gray and Neuhoff(1998)]{RMG-DLN:98}
Robert~M. Gray and David~L. Neuhoff.
\newblock Quantization.
\newblock \emph{IEEE transactions on information theory}, 44\penalty0
  (6):\penalty0 2325--2383, 1998.

\bibitem[Hauenstein et~al.(2011)Hauenstein, Sommese, and Wampler]{Regeneration}
Jonathan~D. Hauenstein, Andrew~J. Sommese, and Charles~W. Wampler.
\newblock Regeneration homotopies for solving systems of polynomials.
\newblock \emph{Math. Comp.}, 80\penalty0 (273):\penalty0 345--377, 2011.

\bibitem[Kwok and Mart{\'\i}nez(2010)]{AK-SM:10}
Andrew Kwok and Sonia Mart{\'\i}nez.
\newblock A coverage algorithm for drifters in a river environment.
\newblock In \emph{American Control Conference (ACC), 2010}, pages 6436--6441.
  IEEE, 2010.

\bibitem[Mart{\'\i}nez and Bullo(2006)]{SM-FB:04p}
Sonia Mart{\'\i}nez and Francesco Bullo.
\newblock Optimal sensor placement and motion coordination for target tracking.
\newblock \emph{Automatica}, 42\penalty0 (4):\penalty0 661--668, 2006.

\bibitem[Mehta et~al.(2015)Mehta, Daleo, D\"orfler, and
  Hauenstein]{KuramotoModel}
Dhagash Mehta, Noah~S. Daleo, Florian D\"orfler, and Jonathan~D. Hauenstein.
\newblock Algebraic geometrization of the {K}uramoto model: equilibria and
  stability analysis.
\newblock \emph{Chaos}, 25\penalty0 (5):\penalty0 053103, 7, 2015.

\bibitem[Schwager et~al.(2009)Schwager, Rus, and Slotine]{MS-DR-JJS:08}
Mac Schwager, Daniela Rus, and Jean-Jacques Slotine.
\newblock Decentralized, adaptive coverage control for networked robots.
\newblock \emph{The International Journal of Robotics Research}, 28\penalty0
  (3):\penalty0 357--375, 2009.

\bibitem[Sommese and Wampler(2005)]{SWbook}
Andrew~J. Sommese and Charles~W. Wampler, II.
\newblock \emph{The numerical solution of systems of polynomials arising in
  engineering and science}.
\newblock World Scientific Publishing Co. Pte. Ltd., Hackensack, NJ, 2005.

\bibitem[Szechtman et~al.(2008)Szechtman, Kress, Lin, and Cfir]{RS-MK-KL-DC:08}
Roberto Szechtman, Moshe Kress, Kyle Lin, and Dolev Cfir.
\newblock Models of sensor operations for border surveillance.
\newblock \emph{Naval Research Logistics (NRL)}, 55\penalty0 (1):\penalty0
  27--41, 2008.

\bibitem[Zemel(1985)]{EZ:85}
Eitan Zemel.
\newblock Probabilistic analysis of geometric location problems.
\newblock \emph{SIAM Journal on Algebraic Discrete Methods}, 6\penalty0
  (2):\penalty0 189--200, 1985.

\end{thebibliography}


\end{document}